\documentclass[conference]{IEEEtran}
\usepackage{amssymb}

\makeatletter
\def\ps@headings{%
\def\@oddhead{\mbox{}\scriptsize\rightmark \hfil \thepage}%
\def\@evenhead{\scriptsize\thepage \hfil \leftmark\mbox{}}%
\def\@oddfoot{}%
\def\@evenfoot{}}
\makeatother

\usepackage{amsfonts}
\usepackage{amssymb,color}
\usepackage{amsmath,epsfig}
\usepackage{rawfonts,graphicx,amsfonts,amssymb}
\usepackage{array}
\usepackage{mathrsfs}
\usepackage{cite}
\usepackage{multirow}

\ifodd 0
\newcommand{\com}[1]{\textbf{\color{blue} (COMMENT: #1)}} 
\else

\newcommand{\com}[1]{}
\fi

\newtheorem{theorem}{Theorem}

\newtheorem{definition}{Definition}

\IEEEoverridecommandlockouts

\ifCLASSINFOpdf \else \fi

\begin{document}

\title{Safe Carrier Sensing Range in CSMA Network under Physical Interference Model}

\author{\IEEEauthorblockN{Liqun Fu, Soung Chang Liew, Jianwei
Huang}
\IEEEauthorblockA{Department of Information Engineering\\
The Chinese University of Hong Kong\\
Shatin, New Territories, Hong Kong\\
Email: \{lqfu6,soung,jwhuang\}@ie.cuhk.edu.hk}

}

\maketitle

\begin{abstract}

In this paper, we study the setting of carrier-sensing range in
802.11 networks under the (cumulative) physical interference model.
Specif\/ically, we identify a carrier-sensing range that will
prevent collisions in 802.11 networks due to carrier-sensing failure
under the physical interference model. We f\/ind that the
carrier-sensing range required under the physical interference model
must be larger than that required under the protocol (pairwise)
interference model by a multiplicative factor. For example, if the
SINR requirement is $10dB$ and the path-loss exponent is $4$, the
factor is $1.4$. Furthermore, given a fixed path-loss exponent of 4,
the factor increases as the SINR requirement increases. However, the
limit of the factor is $1.84$ as the SINR requirement goes to
infinity.
\end{abstract}


\IEEEpeerreviewmaketitle

\section{Introduction}

In 802.11, carrier sensing is designed to prevent concurrent
transmissions that can corrupt each other. The setting of the
carrier-sensing range is crucial. Too large a carrier-sensing range
will limit spatial reuse, while too small a carrier-sensing range
will fail to prevent collisions.

To date, most studies on the proper setting of the carrier-sensing
range \cite{KXu,LiBin,PCNg,SXu,Vasan} are based on the pairwise
interference model (referred to as the protocol interference model
in \cite{prgupta}). For a link under the pairwise interference
model, the interferences from the other links are considered one by
one. If the interference from each of the other links on the link
concerned does not cause a collision, then it is assumed that there
is no collision overall. In particular, the pairwise interference
model does not take into account the cumulative effects of the
interferences from the other links. The physical interference model
\cite{prgupta}, on the other hand, considers the cumulative
interferences and thus is a more accurate model.

Under the pairwise interference model, hidden-node collision happens
if the transmitters of two links do not carrier sense each other,
but the two links are close enough to interfere with each other.
Ref. \cite{LiBin} established the carrier-sensing range required to
prevent hidden-node collisions in 802.11 networks under the pairwise
interference model. This resulting carrier-sensing range is too
optimistic when the more accurate physical interference model is
considered instead. For a particular link, although the
carrier-sensing range is set large enough with respect to the
interference from each of the other links, the cumulative
interference powers from all the other links may still corrupt the
transmission on the link concerned. Since the collision is not due
to a specific ``hidden'' node, the conventional term of ``hidden
node collision'' does not quite apply to this situation that arises
under the physical interference model. We define a new term
``missed-carrier-sensing collision'' to describe this phenomenon.
Missed-carrier-sensing collisions occur when the transmitter does
not sense the transmissions of the other links, but the cumulative
interference power of all the other concurrent transmissions will
interfere the transmission on this particular link. This paper is
dedicated to the study of the required carrier-sensing range to
prevent missed-carrier-sensing collisions in 802.11 networks under
the physical interference model.

\section{System Model}\label{system}

\subsection{Network and Physical Interference Model}\label{Model}

We represent a wireless network by a set of directed links
$\mathcal{L}=\{l_i, 1\le i \le {\left| \mathcal{L} \right|}\}$. Let
$\mathcal{T}=\{T_i, 1\le i \le {\left| \mathcal{L} \right|}\}$ and
$\mathcal{R}=\{R_i, 1\le i \le {\left| \mathcal{L} \right|}\}$
denote the set of transmitting nodes and the set of receiving nodes,
respectively. A receiver decodes its signal successfully if and only
if the signal-to-interference-plus-noise ratio (SINR) requirement at
the receiver is above a certain threshold. We adopt the physical
interference model, where the interference is the sum total of the
powers the receiver receives from all transmitters except its own
transmitter. We assume that radio signal propagation obeys the
log-distance path model with path loss exponent $\alpha>2$. The path
gain $G(T_i,R_j)$ from transmitting node $T_i$ to receiving node
$R_j$ follows a geometric model
\begin{equation}
G(T_i,R_j)={d(T_i,R_j)}^{-\alpha}, \label{channelm}
\end{equation}
where $d(T_i,R_j)$ is the Euclidean distance between nodes $T_i$ and
$R_j$.

In 802.11, each data transfer on a link $l_i$ consists of a DATA
frame in the forward direction (from transmitter $T_i$ to receiver
$R_i$) followed by an ACK frame on the reverse direction (from $R_i$
to $T_i$). The data transfer on link $l_i$ is said to be successful
if and only if both the DATA frame and the ACK frame are correctly
received. So under the physical interference model, the conditions
for successful transmissions on link $l_i$ are
\begin{equation} \frac{P\cdot
G(T_i ,R_i )}{N + \sum\limits_{l_j \in \mathcal{L}'}{P\cdot G(S_j
,R_i )} } \ge {\gamma}_0, \quad \text{(DATA frame)}\label{SINR1}
\end{equation}
and
\begin{equation} \frac{P\cdot
G(R_i ,T_i )}{N + \sum\limits_{l_j \in \mathcal{L}''}{P\cdot G(S_j
,T_i )} } \ge {\gamma}_0, \quad \text{(ACK frame)}\label{SINR2}
\end{equation}
where $P$ is the transmission power level, $N$ is the average noise
power, and $\gamma_0$ is the SINR threshold for correct reception.
We assume that all nodes use the same transmit power $P$ and adopt
the same SINR threshold $\gamma_0$. Let $\mathcal{L}'$
($\mathcal{L}''$) denote the set of links that transmit concurrently
with the DATA (ACK) frame on link $l_i$. For a link $l_j$ in
$\mathcal{L}'$ or $\mathcal{L}''$, the interference could be either
from transmitter $T_j$ or the receiver $R_j$ of link $l_j$ through
the DATA or ACK transmission on link $l_j$, respectively. So we use
the notation $S_j$ to denote the sender of link $l_j$, which could
be either $T_j$ or $R_j$.

\subsection{Carrier Sensing in 802.11}

Consider the wireless link set $\mathcal{L}$. If there exists a link
$l_i\in\mathcal{L}$ such that not both \eqref{SINR1} and
\eqref{SINR2} are satisf\/ied, collision happens. In 802.11, carrier
sensing can be used to prevent simultaneous transmissions that
collide.

We assume carrier sensing by energy detection. That is, if the power
received from another node is above a power threshold $P_{cs}$, then
a transmitter will not transmit and its backoff countdown process
will be frozen. Given a carrier sensing power threshold $P_{cs}$, it
can be mapped to a carrier-sensing range \emph{CSRange}. Consider
two links, $l_i$ and $l_j$. If the transmitters $T_i$ and $T_j$ can
carrier-sense the frames transmitted by each other, simultaneous
transmissions by them will be prevented. That is, if the distance
between $T_i$ and $T_j$ satisf\/ies
\begin{equation}
d(T_i,T_j)<CSRange, \label{eqCSRange}
\end{equation}
the DATA frame transmissions on $l_i$ and $l_j$ are prevented
beforehand.

Setting an appropriate \emph{CSRange} is crucial to the performance
of 802.11 network. If \emph{CSRange} is too large, spatial reuse
will be unnecessarily limited. If \emph{CSRange} is not large
enough, the missed-carrier-sensing collisions may occur. That is a
number of transmitters may transmit simultaneously because condition
\eqref{eqCSRange} is not satisfied by all pairs of the transmitters.
However, there may exist one link that not both conditions
\eqref{SINR1} and \eqref{SINR2} are satisf\/ied. In this case,
collisions happen and the carrier sensing fails to prevent such
collisions. We def\/ine a Safe \emph{CSRange} that will prevent the
missed-carrier-sensing collisions in 802.11 network under the
physical interference model as follows:
\begin{definition}[Safe-CSRange]
Consider the wireless link set $\mathcal{L}$. Let
$\mathcal{L}_{con}\subseteq \mathcal{L}$ denote a subset of links
that are allowed to transmit concurrently under a \emph{CSRange}
setting. Let $\mathcal{L}_{C}(\text{\emph{CSRange}})
=\{\mathcal{L}_{con}\}$ denote all such subsets of links. A
\emph{CSRange} is said to be a \emph{Safe-CSRange} if for any
$\mathcal{L}_{con}\in \mathcal{L}_{C}(\text{\emph{CSRange}})$ such
that for any link $l_i\in\mathcal{L}_{con}$, both the conditions
\eqref{SINR1} and \eqref{SINR2}, with
$\mathcal{L}'=\mathcal{L}''=\mathcal{L}_{con}\setminus \{l_i\}$, are
satisf\/ied.
\end{definition}

In the following analysis, we assume that the background noise power
$N$ can be ignored.

\section{Sufficient Condition on \emph{Safe-CSRange}}

In this section, we will derive a suff\/icient large value on the
\emph{Safe-CSRange} that will prevent the missed-carrier-sensing
collisions in 802.11.

It is shown in \cite{LiBin} that although the \emph{CSRange} is
suff\/icient large for the SINR requirements of all nodes,
transmission failures can still occur due to the ``Receiver-Capture
effect''. Consider two links $l_i$ and $l_j$ such that $T_i$ and
$T_j$ are out of the \emph{CSRange}, but $R_j$ is in the
\emph{CSRange} of $T_i$. If $T_i$ transmits f\/irst, then $R_j$ will
have sensed the signal of $T_i$ and the default operation in most
802.11 products is that $R_j$ will not attempt to receive the later
signal from $T_j$, even if the signal from $T_j$ is stronger. This
will cause the transmission on link $l_j$ to fail. It is further
shown in \cite{LiBin} that no matter how large the \emph{CSRange}
is, we can always come up with an example that give rise to
transmission failures, if the ``Receiver-Capture effect'' is not
dealt with properly. This kind of collisions can be solved with a
receiver ``RS(Re-Start) mode'' which can be enabled in some 802.11
products (e.g., Atheros WiFi chips). With RS mode, a receiver will
switch to receive the stronger packet as long as the SINR threshold
${\gamma}_0$ for the later link can be satisf\/ied.

With the receiver's RS mode, we can derive the \emph{Safe-CSRange}
that will prevent collisions in 802.11 network under the Physical
Interference Model.

\begin{theorem}
\label{SafeCSRange} Consider the wireless link set $\mathcal{L}$.
The suff\/icient condition on the \emph{Safe-CSRange} that will
prevent collisions in 802.11 network under the Physical Interference
Model is:
\begin{equation}
\text{\emph{Safe-CSRange}}=(K+2)d_{\max} \label{SafeCSRangeReq},
\end{equation}
where $d_{\max}=\mathop {\max }\limits_{l_i \in
\mathcal{L}}d(T_i,R_i)$, which is the maximum link length in the
network, and
\begin{equation}
K = \left( {6\gamma_0 \left( {1 + \left( {\frac{2}{\sqrt 3 }}
\right)^\alpha \frac{1}{\alpha - 2}} \right)}
\right)^{\frac{1}{\alpha }}. \label{KReq}
\end{equation}
\end{theorem}

\begin{proof}
With the receiver's RS mode, in order to prevent collisions in
802.11 networks, we only need to show that condition
\eqref{SafeCSRangeReq} is suff\/icient to satisfy both the SIR
requirements \eqref{SINR1} and \eqref{SINR2} of all the concurrent
transmission links.

With the transmitter-side carrier sensing in 802.11, concurrent
transmissions can only happen when the transmitters do not carrier
sense each other. Let $\mathcal{L}_{con}$ denote the set of links
which have concurrent transmissions. Consider any two links $l_i$
and $l_j$ in $\mathcal{L}_{con}$, we have the following inequality:
\begin{equation}
d(T_j,T_i)\geq\text{\emph{Safe-CSRange}}=(K+2)d_{\max}.
\end{equation}
Because both the lengths of links $l_i$ and $l_j$ satisfy
\begin{equation}
d(T_i,R_i)\leq d_{\max}, \quad d(T_j,R_j)\leq d_{\max},\nonumber
\end{equation}
using triangular inequality, we have
\begin{align}
d(T_j,R_i)&\geq d(T_j,T_i)-d(T_i,R_i)\nonumber\\
&\geq (K+2)d_{\max }-d_{\max }=(K+1)d_{\max },
\end{align}

\begin{align}
d(R_j,T_i)&\geq d(T_i,T_j)-d(T_j,R_j)\nonumber\\
&\geq (K+2)d_{\max }-d_{\max }=(K+1)d_{\max },
\end{align}
and also
\begin{align}
d(R_j,R_i)&\geq d(R_i,T_j)-d(T_j,R_j)\nonumber\\
&\geq (K+1)d_{\max }-d_{\max }=Kd_{\max }.
\end{align}

We take the most conservative distance $Kd_{\max }$ in our
interference analysis (i.e., we will pack the interference links in
a tightest manner given the \emph{CSRange} in
\eqref{SafeCSRangeReq}). Consider any two links $l_i$ and $l_j$ in
$\mathcal{L}_{con}$. The following four inequalities are
satisf\/ied:
\begin{align}
d(T_i,T_j)&\geq Kd_{\max },\\
d(T_i,R_j)&\geq Kd_{\max },\\
d(T_j,R_i)&\geq Kd_{\max },\\
d(R_i,R_j)&\geq Kd_{\max }.
\end{align}

Consider any link $l_i$ in $\mathcal{L}_{con}$. We will show that
the SIR requirements for both the DATA frame and the ACK frame can
be satisf\/ied. We first consider the SIR requirement of the DATA
frame. The SIR at $R_i$ is:
\begin{align}
SIR=\frac{Pd^{ - \alpha }\left( {T_i ,R_i }
\right)}{\sum\limits_{l_j\in \mathcal{L}_C,j \ne i} {Pd^{ - \alpha
}\left( {S_j ,R_i } \right)} }
\end{align}

For the received signal power we consider the worst case that
$d(T_i,R_i)=d_{\max }$.  So we have
\begin{equation}
Pd^{ - \alpha }\left( {T_i ,R_i } \right)\geq P\cdot{d_{\max }^{-
\alpha }}.\label{sigP}
\end{equation}

To calculate the cumulative interference power, we consider the
worst case that all the other concurrent transmission links have the
densest packing, in which the link lengths of all the other
concurrent transmission links are equal to zero. In this case, the
links degenerates to nodes. The minimum distance between between any
two links in $\mathcal{L}_{con}$ is $Kd_{\max }$. The densest
packing of nodes with minimum distance requirement is the hexagon
packing (as shown in Fig. \ref{topo}).

\begin{figure}[t]
\begin{center}
\includegraphics [height=5.8cm]{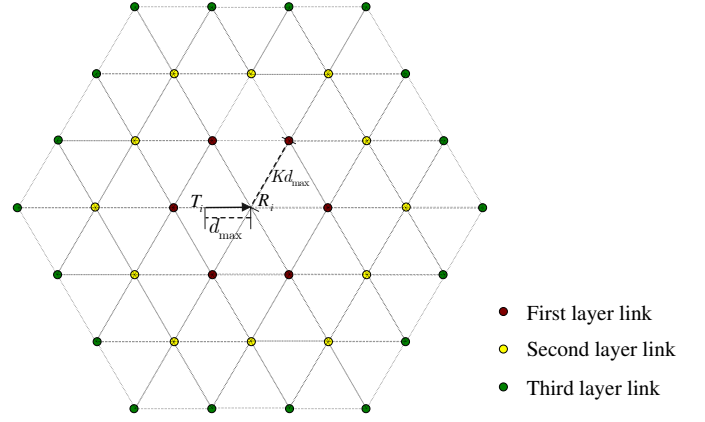}
\end{center}
\begin{center}
\caption{The packing of the interfering links in the worst case}
\label{topo}
\end{center}
\end{figure}

If link $l_j$ is the f\/irst layer neighbor link of link $l_i$, we
have $d(S_j,R_i)\geq Kd_{\max }$. Thus we have
\begin{equation*}
{Pd^{ - \alpha }\left( {S_j ,R_i } \right)}  \leq {P{\left(Kd_{\max
}\right)}^{ - \alpha }}  = \frac{1}{K^{\alpha}}\cdot P{d_{\max }^{-
\alpha }},
\end{equation*}
and there are at most 6 neighbor links in the first layer.

If link $l_j$ is the second layer neighbor link of link $l_i$, we
have $d(S_j,R_i)\geq \sqrt{3} Kd_{\max }$. Thus we have
\begin{equation*}
{Pd^{ - \alpha }\left( {S_j ,R_i } \right)}  \leq {P{\left(\sqrt{3}
Kd_{\max }\right)}^{ - \alpha }}  = \frac{1}{\left(\sqrt{3}
K\right)^{\alpha}}  P{d_{\max }^{- \alpha }},
\end{equation*}
and there are at most 12 neighbor links in the second layer.

If link $l_j$ is the $n$th layer neighbor link of link $l_i$, we
have $d(S_j,R_i)\geq \frac{\sqrt{3}}{2}n\cdot Kd_{\max }$. Thus we
have
\begin{equation*}
{Pd^{ - \alpha }\left( {S_j ,R_i } \right)}  \leq
{P{\left(\frac{\sqrt{3}}{2}n Kd_{\max }\right)}^{ - \alpha }}
=\frac{1}{\left(\frac{\sqrt{3}}{2}n K\right)^{\alpha}} P{d_{\max
}^{- \alpha }},
\end{equation*}
and there are at most $6n$ neighbor links in the second layer.

So the cumulative interference power satisf\/ies the following
inequality:
\begin{align}
&{\sum\limits_{l_j\in \mathcal{L}_C,j \ne i} {Pd^{ - \alpha }\left(
{S_j ,R_i } \right)} }\nonumber\\
\leq &\left(6 \cdot \left( {\frac{1}{K}} \right)^\alpha \mbox{ +
}\sum\limits_{\mbox{n = 2}}^\infty {6n\left( {\frac{2}{\sqrt 3 nK}}
\right)^\alpha }\right)\cdot P{d_{\max }^{- \alpha }} \nonumber\\
= & 6 \cdot \left( {\frac{1}{K}} \right)^\alpha \left( {\mbox{1 +
}\sum\limits_{\mbox{n = 2}}^\infty {n\left( {\frac{2}{\sqrt 3 n}}
\right)^\alpha } } \right)\cdot P{d_{\max }^{- \alpha }}\nonumber\\
= &6 \cdot \left( {\frac{1}{K}} \right)^\alpha \left( {\mbox{1 +
}\left( {\frac{2}{\sqrt 3 }} \right)^\alpha \sum\limits_{\mbox{n =
2}}^\infty {n\left( {\frac{1}{n}} \right)^\alpha } } \right)\cdot
P{d_{\max }^{- \alpha }}\nonumber\\
= &6 \cdot \left( {\frac{1}{K}} \right)^\alpha \left( {\mbox{1 +
}\left( {\frac{2}{\sqrt 3 }} \right)^\alpha \sum\limits_{\mbox{n =
2}}^\infty {\frac{1}{n^{\alpha - 1}}} } \right)\cdot P{d_{\max }^{-
\alpha }}\nonumber\\
\le & 6 \cdot \left( {\frac{1}{K}} \right)^\alpha \left( {\mbox{1 +
}\left( {\frac{2}{\sqrt 3 }} \right)^\alpha \frac{1}{\alpha - 2}}
\right)\cdot P{d_{\max }^{- \alpha }}\label{boundrie}\\
=&\frac{P{d_{\max }^{- \alpha }}}{\gamma_0},\label{follK}
\end{align}
where \eqref{boundrie} follows from a bound on Riemann's zeta
function and \eqref{follK} follows from the definition of $K$ in
\eqref{KReq}.

According to \eqref{sigP} and \eqref {follK}, we find that the SIR
of the DATA frame of link $l_i$ at the receiver $R_i$ satisf\/ies:
\begin{equation}
SIR=\frac{Pd^{ - \alpha }\left( {T_i ,R_i }
\right)}{\sum\limits_{l_j\in \mathcal{L}_C,j \ne i} {Pd^{ - \alpha
}\left( {S_j ,R_i } \right)} } \geq \frac{P\cdot{d_{\max }^{- \alpha
}}}{\frac{P{d_{\max }^{- \alpha }}}{\gamma_0}}=\gamma_0.
\end{equation}
This means the SIR requirement of the successful transmission of the
DATA frame on link $l_i$ can be satisf\/ied.

The proofs that the SIR requirement of the ACK frame on link $l_i$
can be satisf\/ied follow the similar procedure as above. So for any
link $l_i$ in the concurrent transmission links $\mathcal{L}_{con}$,
condition \eqref{SafeCSRangeReq} is suff\/icient to satisfy the SIR
requirements of the successful transmissions of both its DATA and
ACK frames. This means that, together with the receiver's RS mode,
condition \eqref {SafeCSRangeReq} is suff\/icient for preventing
collisions in 802.11 networks under the physical interference model.
\end{proof}

Condition \eqref{SafeCSRangeReq} provides a suff\/iciently large
\emph{CSRange} that prevents missed-carrier-sensing collisions in
802.11 networks. So there is no need to set a \emph{CSRange} larger
than \eqref{SafeCSRangeReq} in order to prevent collisions. Setting
a larger \emph{CSRange} than \eqref{SafeCSRangeReq} will only
decrease spatial reuse.

\section{Comparison of the \emph{Safe-CSRange} with the Pairwise Interference
Model}

This section compares the \emph{Safe-CSRange} under the physical
interference model to that under the pairwise interference model
(derived in \cite{LiBin}). Let
\emph{Safe-CSRange}$_\text{\emph{pairwise}}$ denote the
\emph{CSRange} that prevents hidden-node collisions under the
pairwise interference model. Let
\emph{Safe-CSRange}$_\text{\emph{physical}}$ denote the
\emph{CSRange} that prevents missed-carrier-sensing collisions under
the physical interference model. The
\emph{Safe-CSRange}$_\text{\emph{pairwise}}$ and the
\emph{Safe-CSRange}$_\text{\emph{physical}}$ are:
\begin{align}
&\text{\emph{Safe-CSRange}}_\text{\emph{pairwise}}=\left(2+
{\gamma_0}^{\frac{1}{\alpha }}\right)d_{\max}
\label{SafeCSRangePro},\\
&\text{\emph{Safe-CSRange}}_\text{\emph{physical}}=(K+2)d_{\max}\nonumber\\
&=\left(2+ \left( {6\gamma_0 \left( {1 + \left( {\frac{2}{\sqrt 3 }}
\right)^\alpha \frac{1}{\alpha - 2}} \right)}
\right)^{\frac{1}{\alpha }}\right)d_{\max} \label{SafeCSRangePhy}.
\end{align}

For example, if $\gamma_0=10$ and $\alpha=4$, which are typical for
wireless communications, we have
\begin{align}
&\text{\emph{Safe-CSRange}}_\text{\emph{pairwise}}=3.78d_{\max}
\label{SafeCSRangePro104},\\
&\text{\emph{Safe-CSRange}}_\text{\emph{physical}}=5.27d_{\max}
\label{SafeCSRangePhy104}.
\end{align}
The \emph{Safe-CSRange} needs to be increased by a factor of $1.4$
to ensure successful transmissions under the physical interference
model.

Given a fixed path loss exponent $\alpha$, both
\emph{Safe-CSRange}$_\text{\emph{pairwise}}$ and
\emph{Safe-CSRange}$_\text{\emph{physical}}$ will increase when the
SIR requirement $\gamma_0$ increases. This is the case when the
separation among links must be larger to meet the SIR targets. For
example, if $\alpha=4$, we have
\begin{align}
&\text{\emph{Safe-CSRange}}_\text{\emph{pairwise}}=\left(2+\gamma_0^{\frac{1}{4}}\right)d_{\max}
\label{SafeCSRangePro104},\\
&\text{\emph{Safe-CSRange}}_\text{\emph{physical}}=\left(2+\left(\frac{34}{3}\gamma_0\right)^{\frac{1}{4}}\right)d_{\max}
\label{SafeCSRangePhy104}.
\end{align}
The ratio of \emph{Safe-CSRange}$_\text{\emph{physical}}$ to
\emph{Safe-CSRange}$_\text{\emph{pairwise}}$ is
\begin{equation}
\mathop
\frac{\text{\emph{Safe-CSRange}}_\text{\emph{physical}}}{\text{\emph{Safe-CSRange}}_\text{\emph{pairwise}}}
=  \frac{2 + \left( {\frac{34}{3}\gamma_0 } \right)^{\frac{1}{4}}}{2
+ \gamma_0 ^{\frac{1}{4}}},
\end{equation}
which is also an increasing function of $\gamma_0$. And the limit of
the ratio of \emph{Safe-CSRange}$_\text{\emph{physical}}$ to
\emph{Safe-CSRange}$_\text{\emph{pairwise}}$ as $\gamma_0$ goes to
inf\/inity is:
\begin{align}
\mathop {\lim }\limits_{\gamma_0 \to \infty }
\frac{\text{\emph{Safe-CSRange}}_\text{\emph{physical}}}{\text{\emph{Safe-CSRange}}_\text{\emph{pairwise}}}
&= \mathop {\lim }\limits_{\gamma_0 \to \infty } \frac{2 + \left(
{\frac{34}{3}\gamma_0 } \right)^{\frac{1}{4}}}{2 + \gamma_0
^{\frac{1}{4}}}\nonumber \\&= \left( {\frac{34}{3}}
\right)^{\frac{1}{4}} \approx 1.8348
\end{align}

Fig. \ref{CSRangeratio} shows the ratio
$\frac{\text{\emph{Safe-CSRange}}_\text{\emph{physical}}}{\text{\emph{Safe-CSRange}}_\text{\emph{pairwise}}}$
as a function of the SIR requirements $\gamma_0$. Different curves
represent different choices of path loss exponent $\alpha$.

\begin{figure}[t]
\begin{center}
\includegraphics [height=7.1cm]{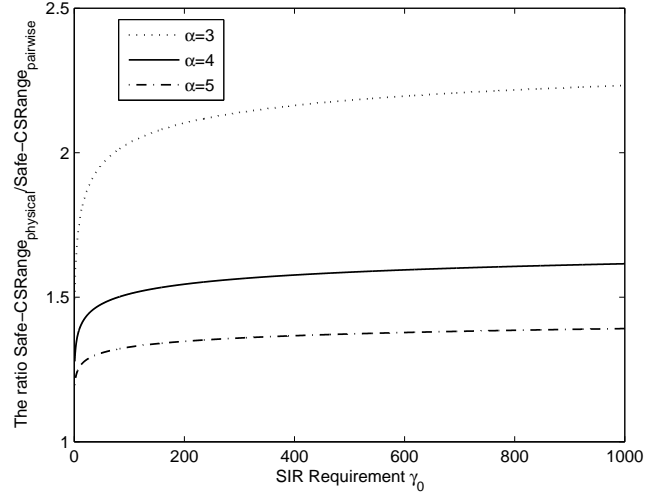}
\end{center}
\begin{center}
\caption{Analytical results of the ratio of the \emph{Safe-CSRange}
under the physical interference model to that under the pairwise
interference model} \label{CSRangeratio}
\end{center}
\end{figure}

\section{Conclusion}\label{conclusion}

This paper studies the problem of setting the carrier-sensing range
that will prevent missed-carrier-sensing collisions in 802.11
networks under the physical interference model. We establish a
suff\/icient condition for the carrier sensing range. We call the
minimum carrier sensing range that meets the condition
\emph{Safe-CSRange} and compare it with that established under the
pairwise interference model \cite{LiBin}. We find that for the
typical setting of path-loss exponent $\alpha=4$ and SINR
requirement $\gamma_0=10$, the \emph{Safe-CSRange} needs to be
increased by a factor of $1.4$ under the physical interference
model. We also find that, given a fixed path-loss exponent $\alpha$,
the factor increases when the SINR requirement $\gamma_0$ increases.
And the factor tends to a constant as the SINR requirement
$\gamma_0$ goes to infinity. For example, when the path-loss
exponent $\alpha=4$, the limit of the factor is $1.84$ as the SINR
requirement $\gamma_0$ goes to infinity.

\end{document}